\documentclass[aps,prx,reprint,superscriptaddress]{revtex4-2}

\usepackage{amssymb,amscd, verbatim,amsmath,color, mathrsfs}
\usepackage{amsthm}
\usepackage{graphicx}
\usepackage{turnstile}
\usepackage{float}
\usepackage{hyperref}
\usepackage{cleveref}
\usepackage{appendix}
\usepackage{quantikz}
\usepackage{subcaption}
\usepackage{tikz}
\usepackage{booktabs}

\newtheorem{thm}{Theorem}[section]

\newtheorem{lem}[thm]{Lemma}
\newtheorem{rem}[thm]{Remark}

\newcommand{\GF}{\mathrm{GF}}
\newcommand{\rank}{\mathrm{rank}}

\begin{document}


\title{Structure-Aware Compilation for Scalable Neutral-Atom Quantum Computing}



\author{Dekuan Dong}
\email{dkdong21@m.fudan.edu.cn}
\affiliation{School of Mathematical Sciences, Fudan University}

\author{Fengyu Zou}
\affiliation{School of Mathematical Sciences, Fudan University}

\author{Hengzhun Chen}
\affiliation{School of Mathematical Sciences, Fudan University}

\author{Guorui Zhu}
\affiliation{School of Mathematical Sciences, Fudan University}

\author{Yingzhou Li}
\email{Corresponding author: yingzhouli@fudan.edu.cn}
\affiliation{School of Mathematical Sciences, Fudan University}
\affiliation{Shanghai Key Laboratory for Contemporary Applied Mathematics}


\date{\today}

\begin{abstract}
	We study the compilation of structured quantum gate families on two-dimensional neutral-atom arrays, aiming to reduce addressing and transport overhead under realistic hardware constraints. For single-qubit gates, we exploit the algebraic structures of gate families at the matrix level, enabling efficient rank-one decompositions over appropriate algebraic structures and thereby reducing the number of addressing layers. For controlled-Z (C-Z) gates, we formulate the transport scheduling problem using graph-theoretic models, leading to efficient compilation algorithms under realistic transport constraints. We provide provable performance guarantees for the proposed methods and validate them through extensive numerical experiments. Across representative single-qubit gate families, our methods reduce the number of addressing layers by up to a factor of two compared with naïve row- or column-wise implementations. For C-Z gates, our scheduling strategy reduces the required number of atom transport operations by approximately 50\%. When applied to QAOA circuits for MaxCut, the proposed framework reduces transport cost by more than 30\% on average. These results show that the physical constraints of neutral-atom hardware can be converted into algebraic and graph-theoretic structure, turning a hardware-level scheduling bottleneck into tractable decomposition and coloring problems.
\end{abstract}


\maketitle

\section{Introduction}
Quantum compilation plays a central role in bridging the gap between high-level quantum algorithms and their physical realization on hardware. Over the past decades, substantial progress has been made in compiler design for leading quantum computing platforms, including superconducting circuits \cite{PhysRevA.76.042319,Clarke2008,Arute2019} and trapped ions \cite{PhysRevLett.117.060504,PhysRevLett.128.160504}. For superconducting architectures, compilation efforts primarily focus on qubit mapping and SWAP minimization under fixed, low-degree connectivity graphs, as well as gate scheduling subject to coherence-time  constraints \cite{10.1145/3680291,10755962,9384317}. In contrast, trapped-ion systems benefit from near all-to-all connectivity, shifting the compilation emphasis toward laser-control optimization, motional mode management, and gate parallelism \cite{9138945,bach2025efficientcompilationshuttlingtrappedion}.

The compilation landscape of neutral-atom quantum computers differs fundamentally from that of both superconducting and trapped-ion platforms. By leveraging optical tweezer arrays and Rydberg-mediated interactions, neutral-atom platforms support reconfigurable qubit layouts, long coherence time, and the ability to physically transport atoms during computation \cite{Beugnon2007,Schlosser2011,PhysRevLett.123.170503,Evered2023,Bluvstein_2025}. These features make neutral-atom platforms a promising candidate for scalable quantum computation, enabling native parallel in-place single-qubit operations and flexible implementations of entangling gates through atomic transport. Efficiently transpiling abstract quantum circuits into hardware-compatible instruction sequences on such platforms, however, remains a nontrivial challenge, particularly when achievable parallelism is limited by circuit structure and overhead associated with atomic transport.

In this work, we develop a systematic framework for transpiling quantum circuits on neutral-atom arrays by explicitly exploiting the algebraic structure of the underlying gate sets. For single-qubit gates, including self-inverse gates, Pauli gates, phase gates, and $T$ ($\pi/8$) gates, we reformulate the parallelization problem as a matrix decomposition problem over appropriate algebraic structures. This perspective enables us to derive provable bounds on the number of required operations for each gate family. The resulting optimality guarantees are summarized in \Cref{tab:optimality_summary}.
\begin{table}
	\centering
	\caption{Summary of optimality guarantees for different single-qubit gate families.}
	\label{tab:optimality_summary}
	\begin{ruledtabular}
	\begin{tabular}{lll}
		Gate family & Algebra structure & Optimality \\
		\midrule
		Self-inverse gates & $\mathbb{Z}_2$ & Optimal \\
		Pauli gates & $\left(\mathbb{Z}_2\right)^2$ & $\tfrac{4}{3}$-optimal \\
		Phase gates & $\mathbb{Z}_4$ & $3$-optimal \\
		$\pi/8$ gates & $\mathbb{Z}_8$ & $7$-optimal (conjectured)\\
	\end{tabular}
	\end{ruledtabular}
\end{table}
For two-qubit C-Z gates, we adopt a graph-theoretical formulation that captures the interaction pattern among atoms. By distinguishing between row/column-aligned and non-aligned gates under realistic transport constraints, we design scheduling strategies that significantly reduce the required number of transport operations.


We evaluate the proposed methods through numerical experiments on randomly generated instances. Across all considered single-qubit gate families, our methods reduce the number of addressing layers by up to a factor of two compared with naïve row- or column-wise implementations. For C-Z gates, the proposed scheduling strategy reduces the required number of atom transport operations by approximately 50\%. We further benchmark the complete compilation framework on randomly generated QAOA-MaxCut instances, where it reduces the C-Z transport cost by more than 30\% on average. In addition, small-scale exhaustive tests confirm the theoretical $\tfrac{4}{3}$-optimality for Pauli gates. Runtime benchmarks further show that all compilation procedures complete within one second for the tested array sizes, indicating that the proposed methods are computationally efficient at near-term neutral-atom scales. 

The main contributions of this work are summarized as follows:
\begin{itemize}
	\item We introduce a unified algebraic framework for decomposing and scheduling several families of single-qubit gates on neutral-atom arrays, with provable bounds on the operation count.
	\item We develop efficient scheduling algorithms for C-Z gates under realistic transport constraints, leveraging graph-theoretical tools such as interval graph coloring and the computation of a minimum number of strictly increasing subsequences.
	\item We validate the effectiveness of our approach through numerical simulations, demonstrating substantial reductions in operation counts relative to na\"{i}ve strategies and verifying theoretical guarantees for Pauli gates.
	\item We apply the proposed compilation framework to QAOA circuits for the MaxCut problem and numerically demonstrate significant improvements in compilation efficiency on neutral-atom architectures.
\end{itemize}
Taken together, these results provide practical guidance for compiling quantum algorithms on neutral-atom platforms, helping to bridge the gap between abstract algorithm design and hardware-aware circuit optimization.
\section{Methods}\label{sec:methods}
\subsection{Single-Qubit Gates}
In a neutral-atom quantum computer, single-qubit gates are implemented by driving qubit transitions using externally applied control fields. Under neutral-atom architecture, each single-qubit addressing layer projects multiple horizontal and vertical laser beams onto the atomic array, so that atoms located at the intersections of these beams are simultaneously addressed and undergo identical single-qubit operations. This addressing scheme is illustrated in \Cref{fig:single_gate}.
\begin{figure}
	\centering
	\includegraphics[width=0.8\linewidth]{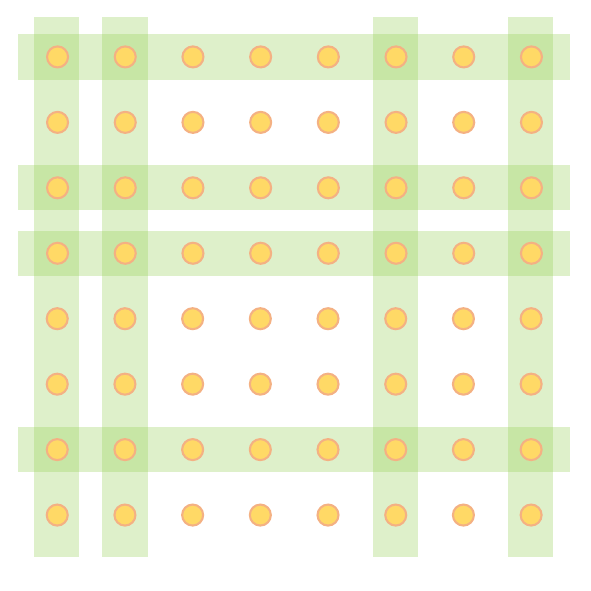}
	\caption{Single-qubit gate addressing. Qubits (yellow circles) are addressed by horizontal and vertical laser beams (semi-transparent green rectangles). Qubits at beam intersections undergo identical single-qubit gates simultaneously.}
	\label{fig:single_gate}
\end{figure}

We focus on several families of single-qubit gates and develop optimized strategies for implementing them on a neutral-atom quantum computer. Consider an $m\times n$ atomic array on which a collection of single-qubit gates from a group $\mathcal G$ is to be applied. This target can be represented by a matrix $M\in \mathcal G^{m\times n}$, where $M_{ij}$ specifies the gate to be applied to the qubit at position $(i, j)$. 

Each single-qubit addressing layer selects a subset of rows and columns and applies a single-qubit gate uniformly to all qubits at their intersections. The goal is to realize the target matrix $M$ with as few such steps as possible. Each step can be expressed as
\[\left(u v^\top\right) \cdot G, \quad G\in \mathcal G, \; u\in \{0, 1\}^m, \; v\in \{0, 1\}^n,\]
where the binary vectors $u$ and $v$ indicate the selected rows and columns, respectively. The outer product $uv^\top$ specifies the subset of addressed qubits, and multiplication by a group element $G$ corresponds to applying $G$ to the selected positions according to
\[0\cdot G = I, \quad 1\cdot G = G, \]
where $I$ denotes the identity element of $\mathcal G$. 

Using this notation, the compilation problem can be formulated as the following optimization problem:
\begin{equation}\label{eq:main_opt}
	\begin{aligned}
		\min \quad & k\\
		\text{s.t.} \quad & M = \sum_{t = 1}^k \left(u_t v_t^\top\right)\cdot G_t,\\
		& G_t\in \mathcal G, \; u_t\in \{0, 1\}^m, \; v_t\in \{0, 1\}^n,
	\end{aligned}
\end{equation}
where the sum is taken element-wise in $\mathcal G^{m\times n}$ using the group operation.
\subsubsection{Self-inverse Gates}\label{sec:self_inverse_gates}
We first consider the case $\mathcal G = \{I, X\}$, where $X$ is a self-inverse single-qubit gate, i.e., $X^2 = I$. Many commonly used single-qubit gates, including the Pauli gates and the Hadamard gates, satisfy this property. In this setting, the group $\mathcal G$ is isomorphic to the finite field $\GF(2)$, and the optimization problem \eqref{eq:main_opt} reduces to 
\[\begin{aligned}
	\min \quad & k\\
	\text{s.t.} \quad & M = \sum_{t = 1}^k u_t v_t^\top,\quad u_t\in \{0, 1\}^m, \; v_t\in \{0, 1\}^n,
\end{aligned}\]
where $M \in \{0, 1\}^{m\times n}$ encodes the target single-qubit gate pattern, and the summation is element-wise in $\GF(2)$.

The reduced problem is equivalent to computing the rank of $M$ over $\GF(2)$. Since $\GF(2)$ is a field, Gaussian elimination can be applied to determine the $\GF(2)$-rank and to explicitly construct a corresponding decomposition. The procedure is as follows:
\begin{enumerate}
	\item Apply a sequence of elementary row operations, represented by an invertible matrix $T_r$, to transform $M$ into reduced row echelon form
	\[T_r M = e_1 v_1^\top + e_2 v_2^\top + \cdots + e_k v_k^\top, \]
	where $e_i$ is the $i$-th column of the $m\times m$ identity matrix, $\{v_1, \dots, v_k\}$ are linearly independent over $\GF(2)$. 
	\item Define $u_t = T_r^{-1} e_t$. Then $M$ admits the decomposition 
	\[M = u_1 v_1^\top + u_2 v_2^\top + \cdots + u_k v_k^\top.\]
\end{enumerate}
The $\GF(2)$-rank of $M$ thus characterizes the minimum number of addressing layers required to realize the target gate pattern under the row-column addressing constraint of the neutral-atom architecture. In contrast, a na\"ive implementation that applies gates sequentially column-by-column or row-by-row would generally require up to $\min\{m, n\}$ steps. This self-inverse scenario provides a baseline, illustrating how algebraic structure can be exploited to optimize the compilation process.
\subsubsection{Pauli Gates}\label{sec:pauli_gates}
We now consider the Pauli group 
\[\left\{\pm I, \pm X, \pm Y, \pm Z, \pm \mathrm{i} I, \pm \mathrm{i} X, \pm \mathrm{i} Y, \pm \mathrm{i} Z\right\}.\]
Since global phases are physically irrelevant, we restrict attention to the phase-free Pauli group $\mathcal G = \{I, X, Y, Z\}$. The group operation, denoted by $+$ and defined by Pauli multiplication modulo global phase, is summarized in the following table:
\begin{table}[H]
	\centering
	\begin{tabular}{c|cccc}
		$+$ &   $I$   &   $X$   &   $Y$   &   $Z$\\ \hline
		$I$ &   $I$   &   $X$   &   $Y$   &   $Z$\\ 
		$X$ &   $X$   &   $I$   &   $Z$   &   $Y$\\ 
		$Y$ &   $Y$   &   $Z$   &   $I$   &   $X$\\ 
		$Z$ &   $Z$   &   $Y$   &   $X$   &   $I$
	\end{tabular}
\end{table}
\noindent Here, the identity $I$ acts as the zero element. The group $\mathcal G$ is Abelian and isomorphic to $(\mathbb Z_2)^2$ via the mapping
\[I \leftrightarrow (0, 0) , \; X \leftrightarrow (0, 1), \; Y \leftrightarrow (1, 0), \; Z \leftrightarrow (1, 1),\]
under which the group operation reduces to component-wise addition modulo $2$, i.e., the bitwise XOR operation.

With this identification, the optimization problem \eqref{eq:main_opt} can be reformulated as
\[\begin{aligned}
	\min \quad & k\\
	\text{s.t.} \quad & M = \sum_{t = 1}^k \left(u_t v_t^\top\right)\cdot G_t,\\
	& G_t\in \{0, 1\}^2, \; u_t\in \{0, 1\}^m, \; v_t\in \{0, 1\}^n,
\end{aligned}\]
which is equivalent to computing the rank and constructing a rank-one decomposition of an $m\times n \times 2$ tensor over $\GF(2)$. For a third-order tensor $\mathcal T = [X_1|X_2|\cdots|X_p]\in \mathbb F^{m\times n \times p}$, where $\mathbb F$ is a given field, the tensor rank is the minimal number $r$ such that each slice $X_i\in \mathbb F^{m\times n}$ can be expressed as a linear combination of $r$ rank-one matrices over $\mathbb F$.

In general, computing the rank of a third-order tensor over a finite field is NP-complete \cite{haastad1989tensor}. Here, the third dimension is fixed to be $2$, for which several special-case results exist \cite{ja1978optimal,song2022maximal,song2025rank}, though the exact complexity remains open.

We present a simple algorithm that produces a rank-one decomposition of an $m\times n \times 2$ tensor over $\GF(2)$ using at most $4/3$ times the optimal number of terms. Let $\mathcal T = [X_1 |X_2]$ be a tensor with two slices. It admits three equivalent decompositions:
\[\begin{aligned}
	\mathcal T &= X_1 \otimes \begin{bmatrix}0\\1\end{bmatrix} + X_2 \otimes \begin{bmatrix}1\\0\end{bmatrix}\\
	&= (X_1 + X_2)\otimes \begin{bmatrix}0\\1\end{bmatrix} + X_2 \otimes \begin{bmatrix}1\\1\end{bmatrix}\\
	&= X_1 \otimes \begin{bmatrix}1\\1\end{bmatrix} + (X_1 + X_2) \otimes \begin{bmatrix}1\\0\end{bmatrix}.  
\end{aligned}\]
Each representation expresses $\mathcal T$ as the sum of two matrix-vector tensor products. Decomposing the matrices into rank-one components gives rank-one expansions with 
\[\begin{aligned}
	&\rank(X_1) + \rank(X_2),\\
	&\rank(X_1 + X_2) + \rank(X_2),\\
	&\rank(X_1) + \rank(X_1 + X_2)
\end{aligned}
\]
terms, respectively. Our algorithm chooses the decomposition with the fewest rank-one terms,
\[\begin{aligned}
	\Delta := &\min \{\rank(X_1) + \rank(X_2), \\
	&\quad \rank(X_1) + \rank(X_1 + X_2), \\
	&\quad \rank(X_2) + \rank(X_1 + X_2)\}.
\end{aligned}\]
Applying Lemma 8 of \cite{song2025rank} (see \Cref{sec:lems&pfs}), 
	\[\rank(\mathcal T) \ge \frac{\rank(X_1) + \rank(X_2) + \rank(X_1 + X_2)}{2},\]
we obtain
\[\begin{aligned}
	\Delta &\le \frac23 \left(\rank(X_1) + \rank(X_2) + \rank(X_1 + X_2)\right)\\
	& \le \frac43 \rank(\mathcal T), 
\end{aligned}\]
establishing that our algorithm produces a rank-one decomposition using at most $4/3$ times the optimal number of terms.
\subsubsection{Phase Gates}\label{sec:phase_gates}
We next consider the cyclic group generated by the phase gate 
\[S = \begin{bmatrix}1& 0\\ 0 & \mathrm{i}\end{bmatrix}, \]
namely $\mathcal G = \{I, S, S^2, S^3\}$. The group has order four and is isomorphic to the additive group modulo $4$, $\mathbb Z_4 = \{0, 1, 2, 3\}$. Under this identification, the optimization problem \eqref{eq:main_opt} becomes
\begin{equation}\label{eq:deco_S}
	\begin{aligned}
		\min \quad & k\\
		\text{s.t.} \quad & M = \sum_{t = 1}^k \left(u_t v_t^\top\right)\cdot G_t,\\
		& G_t\in \mathbb Z_4, \; u_t\in \{0, 1\}^m, \; v_t\in \{0, 1\}^n,
	\end{aligned}
\end{equation}
with all additions performed modulo $4$. We denote the optimal value of $k$ by $r_4(M)$.  

We present a polynomial-time algorithm that produces a decomposition of $M$ using at most $3r_4(M)$ terms. The procedure consists of two stages:
\begin{enumerate}
	\item Let $M_1 = M \bmod 2$. Compute its rank over $\GF(2)$ and obtain a rank-one decomposition
	\[M_1 = \sum_{t=1}^{k_1} u_t^{(1)} \left(v_t^{(1)}\right)^\top,   \]
	where $u_t^{(1)} \in \{0, 1\}^m, \; v_t^{(1)}\in \{0, 1\}^n$.
	\item Define
	\[ M_2 = M - \sum_{t=1}^{k_1} u_t^{(1)} \left(v_t^{(1)}\right)^\top (\bmod 4), \]
	so that $M_2 \in \{0, 2\}^{m\times n}$. Compute the rank of $M_2 / 2$ over $\GF(2)$ and obtain
	\[\frac12 M_2 = \sum_{t=1}^{k_2} u_t^{(2)} \left(v_t^{(2)}\right)^\top,\]
	where $u_t^{(2)} \in \{0, 1\}^m, \; v_t^{(2)}\in \{0, 1\}^n$.
\end{enumerate}
Combining the two stages, we obtain a decomposition
\begin{equation}\label{eq:final_deco}
	\begin{aligned}
		M = \sum_{t=1}^{k_1} u_t^{(1)} \left(v_t^{(1)}\right)^\top + \sum_{t=1}^{k_2}  2 u_t^{(2)} \left(v_t^{(2)}\right)^\top (\bmod 4), 
	\end{aligned} 
\end{equation}
which is valid under the constraints of \eqref{eq:deco_S}.

Since by definition there exists a decomposition of $M$ satisfying the constraints of \eqref{eq:deco_S} using $k = r_4(M)$ terms, reducing each coefficient modulo $2$ produces a valid decomposition of $M_1$ over $\GF(2)$:
\[M_1 = \sum_{t=1}^{r_4(M)}\left(u_t v_t^\top\right) \cdot \left(G_t \bmod 2\right). 
\]
This immediately implies that $k_1 \le r_4(M)$.

To analyze the second stage, it is convenient to introduce an auxiliary quantity
\begin{equation}\label{eq:rho_A}
	\min\left\{k: M = \sum_{t=1}^k u_t v_t^\top \bmod 4, \; u_t\in \mathbb Z_4^m, \; v_t \in \mathbb Z_4^n\right\},
\end{equation}
which is denoted by $\rho_4(M)$. In contrast to $r_4(M)$, it allows multiplication between elements of $\mathbb Z_4$ and thus endows $\mathbb Z_4$ with its standard ring structure. In the following, $\rho_4(M)$ will serve as an intermediate measure to bound the number of rank-one terms produced in the second stage. A relationship between $\rho_4(M)$ and the matrix $\GF(2)$-rank is formalized in the following lemma, whose proof is given in \Cref{sec:lems&pfs}.
\begin{lem}\label{lem:rho_rank}
	Suppose $M\in \{0, 2\}^{m\times n}$ and define $N\in \{0, 1\}^{m\times n}$ by
	\[N_{ij} = \begin{cases}
		1, & M_{ij} = 2,\\
		0, & M_{ij} = 0.
	\end{cases}\]
	Then $\rho_4(M)$ is equal to the rank of $N$ over $\GF(2)$.
\end{lem}
Applying \Cref{lem:rho_rank} to $M_2$ gives $k_2 = \rho_4(M_2)$. By construction, $M_2$ can be expressed as 
\[M_2 = \sum_{t=1}^{r_4(M)} (u_t v_t^\top) \cdot G_t - \sum_{t=1}^{k_1} u_t^{(1)}\left(v_t^{(1)}\right), \]
which is a valid decomposition under the definition of $\rho_4(\cdot)$. It then follows that 
\[k_2 = \rho_4(M_2) \le r_4(M) + k_1 \le 2 r_4(M). \]
so that the total number of terms in the final decomposition \eqref{eq:final_deco} satisfies
\[k_1 + k_2 \le 3 r_4(M).\] 
\begin{rem}
	The vectors $u_t^{(1)}, v_t^{(1)}$ obtained in the first stage provide some flexibility in the second stage. Specifically, one may alternatively define
	\[M_2 = M - \sum_{t=1}^{k_1} \left(u_t^{(1)} \left(v_t^{(1)}\right)^\top\right) \cdot G_t \]
	for any choice $G_t \in \{1, 3\}$, and in practice the coefficients $G_t$ can be selected to minimize $k_2$. 
\end{rem}
\subsubsection{Clifford Gates}
We now consider the single-qubit Clifford group generated by $\{I, X, Y, Z, S, H\}$. Up to an irrelevant global phase, any single-qubit Clifford operator admits a decomposition of the form
\[U = PH^a S^b\]
where $P\in \{I, X, Y, Z\}$, $a\in \{0, 1\}$, and $b\in \{0, 1, 2, 3\}$. 

This factorization suggests a natural compilation strategy for Clifford gates on a neutral-atom array. In particular, the Clifford compilation problem can be decomposed into three independent subproblems corresponding to the implementation of phase gates, Hadamard gates, and Pauli gates. Concretely, we first realize the required phase-gate pattern using the algorithm developed in \Cref{sec:phase_gates}. We then implement the Hadamard gates using the method for self-inverse gates described in \Cref{sec:self_inverse_gates}. Finally, the remaining Pauli gates are synthesized using the approach presented in \Cref{sec:pauli_gates}. 

The modular decomposition enables the Clifford compilation problem to be addressed by combining specialized algorithms for each gate family, each with provable performance guarantees. As a result, the overall compilation procedure is efficient, scalable, and well suited to the row-column addressing constraints of neutral-atom architectures. 
\subsubsection{$\pi / 8$ Gates}
In this subsection, we consider the group generated by the $\frac\pi 8$ gate
\[T = \begin{bmatrix}
	1&0\\ 0 & e^{\mathrm{i} \frac\pi 4}
\end{bmatrix}.\]
Analogous to the phase-gate group, this group is cyclic of order eight and is isomorphic to the additive group modulo $8$, $\mathbb Z_8 = \{0, 1, 2, 3, 4, 5, 6, 7\}$. Under this identification, the compilation problem \eqref{eq:main_opt} reduces to 
\begin{equation}\label{eq:Z_8}
	\begin{aligned}
		\min \quad & k\\
		\text{s.t.} \quad & M = \sum_{t = 1}^k \left(u_t v_t^\top\right)\cdot G_t,\\
		& G_t\in \mathbb Z_8, \; u_t\in \{0, 1\}^m, \; v_t\in \{0, 1\}^n,
	\end{aligned}
\end{equation} 
where all additions are performed modulo $8$. We denote the optimal value of $k$ by $r_8(M)$.

As in the phase-gate case, we adopt a two-stage decomposition strategy:
\begin{enumerate}
	\item Let $M_1 = M \bmod 2$. Compute its rank over $\GF(2)$ and obtain a rank-one decomposition
	\[M_1 = \sum_{t=1}^{k_1} u_t^{(1)} \left(v_t^{(1)}\right)^\top, \]
	where $u_t^{(1)} \in \{0, 1\}^m, \; v_t^{(1)}\in \{0, 1\}^n$.
	\item Define
	\[ M_2 = M - \sum_{t = 1}^{k_1} u_t^{(1)} \left(v_t^{(1)}\right)^\top (\bmod 8), \]
	so that $M_2 \in \{0, 2, 4, 6\}^{m\times n}$. Apply the algorithm for compiling phase gates and obtain
	\[\frac12 M_2 = \sum_{t = 1}^{k_2} \left(u_t^{(2)} \left(v_t^{(2)}\right)^\top\right) \cdot G_t, \]
	where $G_t \in \mathbb Z_4, \;u_t^{(2)} \in \{0, 1\}^m, \; v_t^{(2)}\in \{0, 1\}^n$.
\end{enumerate}
Combining the two stages, we obtain a valid decomposition
\[M = \sum_{t = 1}^{k_1} u_t^{(1)} \left(v_t^{(1)}\right)^\top + \sum_{t = 1}^{k_2} \left(2G_t\right)u_t^{(2)} \left(v_t^{(2)}\right)^\top (\bmod 8).\]
As in the phase gate case, reducing an optimal $\mathbb Z_8$ decomposition modulo $2$ implies $k_1 \le r_8(M)$. Moreover, applying the phase-gate algorithm to $M_2 / 2$ yields $k_2 \le 3 r_4(M_2 / 2)$. We conjecture that $r_4(M_2 / 2) = r_8(M_2)$, which implies 
\[r_4(M_2 / 2) = r_8(M_2) \le r_8(M) + k_1 \le 2 r_8(M).\]
Under this conjecture, the total number of terms produced by the algorithm satisfies $k_1 + k_2 \le 7 r_8(M)$.
\begin{rem}
	The recursive structure of the above algorithm extends naturally to the cyclic group generated by a $\pi/2^{p}$ rotation,
	\[
	R_p=
	\begin{bmatrix}
		1&0\\
		0&e^{i\pi/2^{p}}
	\end{bmatrix},
	\]
	which is isomorphic to $\mathbb Z_{2^{p+1}}$. Under this identification, the compilation problem reduces to an optimization problem analogous to \eqref{eq:Z_8}, with the coefficients satisfying $G_t\in \mathbb Z_{2^{p+1}}$. One can then first reduce the target matrix modulo $2$ and compute a rank-one decomposition over $\mathrm{GF}(2)$, as in the first stage of the phase- and $\pi/8$-gate algorithms. After subtracting this binary decomposition, all entries of the residual matrix become even, so the residual can be divided by two and regarded as an instance over $\mathbb Z_{2^{p}}$. Repeating this procedure recursively reduces the original compilation problem to a sequence of binary matrix decompositions and lower-order cyclic-group compilation problems.
\end{rem}

\subsection{C-Z Gates}
Unlike single-qubit gates, the execution of C-Z gates requires the participating qubits to be physically transported to a designated entangling zone. Within this zone, C-Z gates can be performed between adjacent rows or columns of qubits. An acousto-optic deflector (AOD) is used to extract multiple qubits from the atom array and transport them into the entangling zone. In a single transport, several rows and columns of the atom array are selected, and all atoms located at their intersections are moved simultaneously. During the transport, the relative positions of the selected atoms are preserved \cite{zhu2025quantumcompilerdesignqubit}. The transportation process is illustrated in \Cref{fig:transport}.

\begin{figure}[H]
	\centering
	\includegraphics[width=1.0\linewidth]{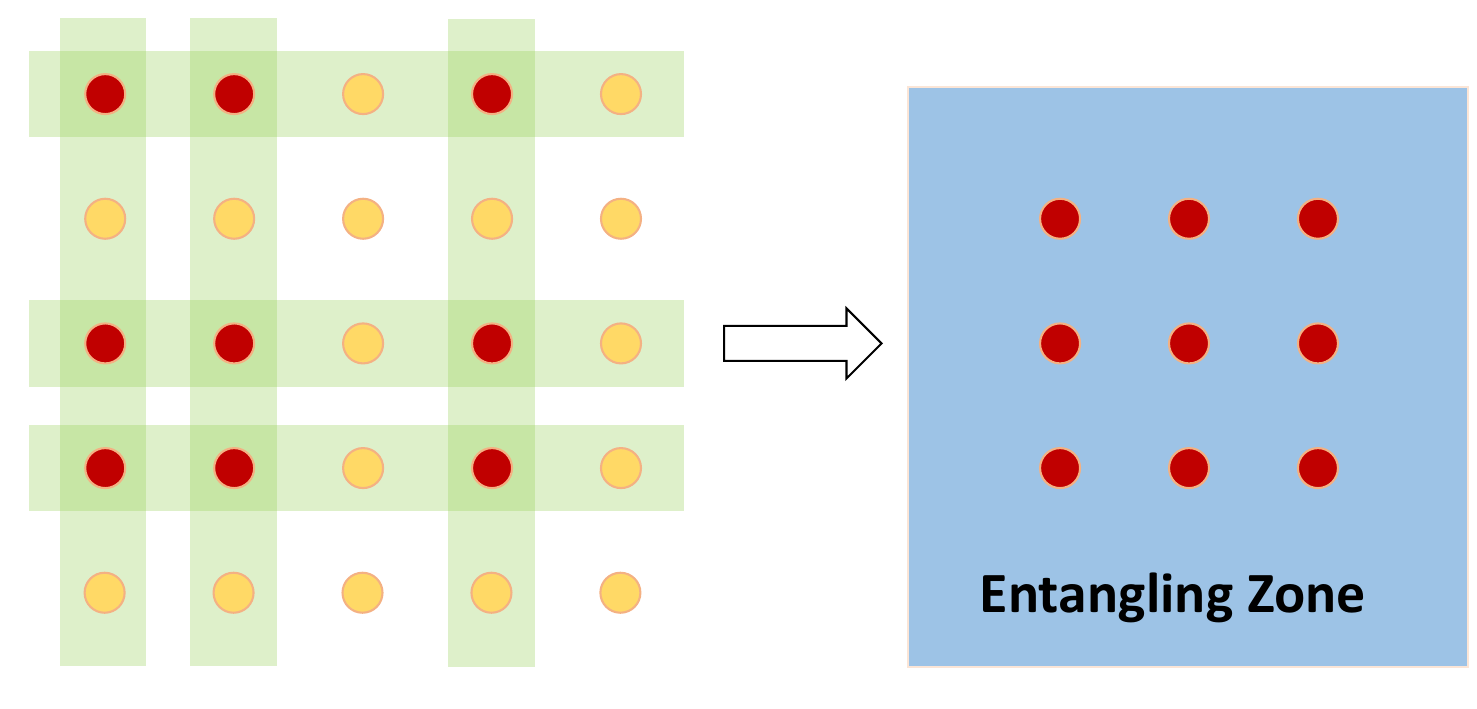}
	\caption{Illustration of qubit transport. In the original neutral-atom array, several qubits (highlighted with darker colors) are selected by AOD and transported to the entangling zone. }
	\label{fig:transport}
\end{figure}

In the entangling zone, C-Z gates are implemented by bringing two adjacent rows (or columns) of atoms sufficiently close such that their separation is smaller than the characteristic Rydberg blockade radius. In this regime, corresponding atoms in the two rows (or columns) can undergo simultaneous C-Z gate operations. Consequently, the following constraints must be satisfied when performing C-Z gates in the entangling zone:
\begin{itemize}
	\item C-Z gates can only be applied between vertically or horizontally adjacent qubits.
	\item C-Z gates must be executed simultaneously on all corresponding atom pairs between two adjacent rows (or columns), rather than on only a subset of those pairs. 
\end{itemize}

Under the transportation and entangling constraints, a C-Z gate between two qubits located in different rows and different columns cannot be implemented within a single transport operation. To realize such C-Z gates, atoms selected from two rows are transported to the entangling zone in two separate steps. The process is illustrated in \Cref{fig:cz-gates}. 

\begin{figure}[H]
	\centering
	\includegraphics[width=1.0\linewidth]{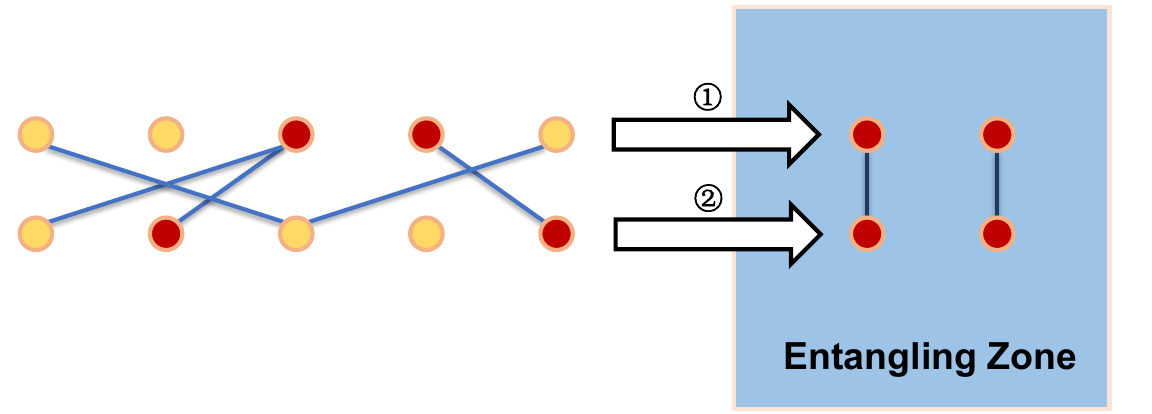}
	\caption{Transporting selected rows of qubits in two steps. After the transport operations, C-Z gates between qubits located in different rows and columns become column-aligned in the entangling zone, enabling their implementation. }
	\label{fig:cz-gates}
\end{figure}

Because qubit transport is both the most time-consuming operation and a dominant source of error in neutral-atom architectures, an essential goal of C-Z gate compilation is to realize a given set of target C-Z gates using as few transport operations as possible.

Under the assumptions described above, C-Z gates naturally divide into two categories: row/column-aligned C-Z gates, which act on pairs of qubits within the same row or the same column, and non-aligned C-Z gates, which couple qubits located in different rows and different columns. We analyze these two cases separately and develop corresponding compilation strategies in the following subsections. By combining these strategies, we ultimately obtain a unified compilation framework capable of realizing arbitrary patterns of C-Z gates on a neutral-atom quantum computer. 
\subsubsection{Row/column-aligned C-Z Gates}\label{sec:intra_CZ}
We consider the compilation problem for row-aligned C-Z gates, namely, implementing a set of C-Z gates in which each gate acts on a pair of qubits within the same row of the atom array, while minimizing the number of transport operations. The compilation problem for column-aligned C-Z gates can be treated analogously. 

We first consider a simpler setting in which all C-Z gates act on qubits within a single row. Suppose the row contains $n$ qubits, labeled $0, 1, \dots, n-1$ from left to right according to their spatial order in the array. Each C-Z gate is uniquely specified by an ordered pair $(a, b)$ with $0\le a < b \le n-1$, indicating that the gate acts on qubits $a$ and $b$. Because the relative ordering of qubits is preserved during transport, two C-Z gates $(a_1, b_1)$ and $(a_2, b_2)$ can be executed within a single transport operation if and only if their corresponding intervals do not overlap, namely
\[(a_1, b_1) \cap (a_2, b_2) = \emptyset, \]
where each gate is identified with the open interval $(a, b)$ on the real line. Under this interpretation, the physical scheduling problem reduces to a purely combinatorial one. Given a set of intervals $\left\{(a_i, b_i)\right\}$, the goal is to partition them into the minimum number of groups such that no two intervals within the same group intersect. From a graph-theoretic perspective, this is precisely the graph-coloring problem on an interval graph, whose vertices correspond to the intervals $(a_i, b_i)$, with edges connecting pairs of overlapping intervals. Since interval graphs are perfect \cite{golumbic2004algorithmic}, their chromatic number, namely, the minimum number of colors required so that adjacent vertices receive different colors, can be computed in polynomial time by sorting the intervals and greedily assigning colors. The resulting coloring directly yields an optimal transport schedule for the corresponding set of C–Z gates.

When the set of row-aligned C-Z gates involves $m$ different rows, a straightforward approach is to implement the gates row by row using the strategy described in the previous paragraph. Here, we present an alternative method that exploits the structure of qubit transport and the self-inverse property of C-Z gates.

Specifically, the row-aligned C-Z gates acting on an $m\times n$ atom array can be represented by an $m\times n(n-1)/2$ binary matrix $M$. The columns of $M$ are indexed by ordered pairs $(a, b)$ with $0\le a < b \le n-1$. The entry $M_{i, (a, b)}$ equals $1$ (respectively, $0$) if a C-Z gate is (respectively, is not) applied between qubits $a$ and $b$ in the $i$-th row. An illustration example is shown in \Cref{fig:to_matrix}.
\begin{figure}[H]
	\centering
	\includegraphics[width=1.0\linewidth]{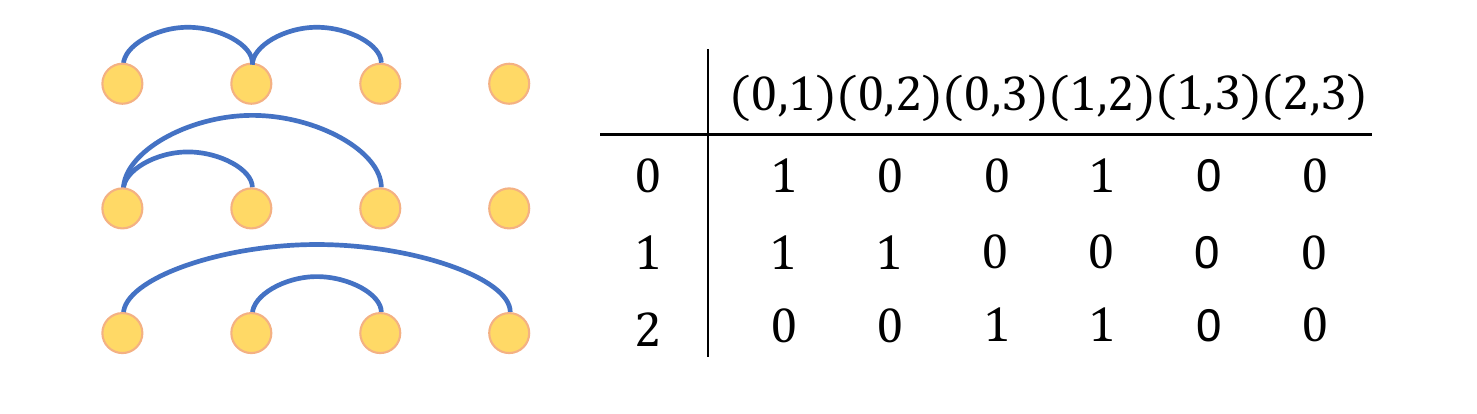}
	\caption{Encoding a set of C-Z gates into a binary matrix representation. }
	\label{fig:to_matrix}
\end{figure}
Then, the matrix is partitioned according to its column indexes, such that any two indexes (interpreted as intervals) within the same group are disjoint. This grouping can be obtained in two steps:
\begin{enumerate}
	\item Collect the indexes of all nonzero columns.
	\item Interpret these column indexes as open intervals and apply interval-graph coloring to obtain a valid grouping.
\end{enumerate}
For example, after removing all-zero columns, the matrix in \Cref{fig:to_matrix} can be grouped as follows:
\[\begin{array}{c||c|c||c||c} 
	\quad \quad & (0,1) & (1,2) & (0,2) & (0,3) \\
	\hline 0 & 1 & 1 & 0 & 0 \\
	\hline 1 & 1 & 0 & 1 & 0 \\
	\hline 2 & 0 & 1 & 0 & 1
\end{array}\]
Note that within each group, any collection of C-Z gates corresponding to a rank-$1$ pattern can be implemented using a single transport operation. Moreover, since the C-Z gate is self-inverse, Gaussian elimination can be applied to determine the minimum number of transport operations required to realize the desired gate set within each group. Combining the implementations for all groups then yields a complete procedure for executing the full set of C-Z gates. 

In the above example, the first group of columns has rank two, whereas the remaining two groups each have rank one. Consequently, the entire set of C-Z gates can be implemented using a total of $4$ transport operations. 
\begin{rem}
	We now present a unified perspective that encompasses both the row-by-row method and the grouping-and-elimination method. Specifically, the row indexes can first be partitioned into $k$ groups, where the $i$-th group contains $m_i$ rows. The grouping-and-elimination procedure is then applied independently to each group.
	
	From this viewpoint, the row-by-row method corresponds to the special case in which all \(m_i=1\), whereas the grouping-and-elimination method corresponds to the case \(k=1\). In practice, heuristic strategies can be developed to optimize the grouping of both rows and columns of the matrix \(M\), thereby reducing the overall transport cost.   
\end{rem}
\subsubsection{Non-aligned C-Z Gates}\label{sec:inter_CZ}
We consider the problem of implementing non-aligned C-Z gates, namely, C-Z gates acting on pairs of qubits located in different rows and columns of the atom array, with the goal of minimizing the number of transport operations. 

We first consider the case of two rows of qubits, each containing $n$ qubits labeled $1, 2, \dots, n$ from left to right. An inter-row C-Z gate is uniquely specified by a pair $(a_1, a_2)$, indicating that the gate acts on qubit $a_1$ in the first row and qubit $a_2$ in the second row. Since the relative ordering of qubits within each row is preserved during transport, two inter-row C-Z gates $(a_1, a_2)$ and $(b_1, b_2)$ can be executed within a single transport operation if and only if their induced connections do not cross, i.e., 
\[(a_1 - b_1)(a_2 - b_2) > 0. \] 
Under this condition, the scheduling problem admits a combinatorial reformulation. Given a set of pairs $S \subset \mathbb Z_+^2$, the task is to find the smallest integer $k$ together with a partition $\{S_1, \dots , S_k\}$ of $S$ such that 
\[(a_1 - b_1)(a_2 - b_2) > 0, \; \forall (a_i, b_i)\in S_j, \;i = 1, 2,\; j = 1, \dots, k.\]
Let $S = \{a^{(1)}, \dots, a^{(m)}\}$ with $a^{(i)} = \left(a_1^{(i)}, a_2^{(i)}\right)$. We sort these points by non-decreasing order of $a_1^{(i)}$, breaking ties by decreasing $a_2^{(i)}$, and define the resulting sequence
\[x_i = a_2^{(i)}, \quad i = 1, \cdots, m.\]
Under this ordering, each subset $S_j$ corresponds to a strictly increasing subsequence of $(x_{1}, \dots, x_{m})$. Consequently, the minimum number of transport operations is equal to the minimum number of strictly increasing subsequences required to partition the sequence $(x_1, \dots, x_m)$. This quantity can be computed optimally using a greedy algorithm with time complexity $\mathcal O(m\log m)$.

Consequently, a straightforward strategy for compiling non-aligned C-Z gates is to process the gates row pair by row pair. Specifically, for each pair of rows, all required inter-row C-Z gates are implemented using the method described above.
\subsubsection{General C-Z Gates Scheduling}
By combining the two cases discussed above, arbitrary C-Z gates on a neutral-atom array can be scheduled efficiently. 

First, all non-aligned C-Z gates are scheduled row pair by row pair using the method described in \Cref{sec:inter_CZ}. During this process, some row- or column-aligned C-Z gates may also be executed without requiring additional transport operations. The remaining row-aligned and column-aligned gates are then scheduled separately using the approach developed in \Cref{sec:intra_CZ}. 

\section{Application: Compilation of QAOA-MaxCut}
The Quantum Approximation Optimization Algorithm (QAOA) circuit for the MaxCut problem is a $p$-layer parameterized quantum-classical hybrid circuit. Each layer consists of a fixed gate sequence, with rotation angles determined by a parameter pair $(\gamma, \beta)$ that governs the controlled-phase ($Z_\gamma$) and $R_x$ gates. An example graph and the corresponding single-layer QAOA-MaxCut circuit are shown in  \Cref{fig:eg_QAOA}.
\begin{figure*}
	\centering
	\begin{subfigure}[c]{0.45\textwidth}
		\centering
		\begin{tikzpicture}[
			scale=1.2,
			every node/.style={
				draw,
				circle,
				fill=white,
				minimum size=6mm,
				inner sep=1pt
			}
			]
			\node (3) at (0.0,1.7) {3};
			\node (4) at (1.7,1.5) {4};
			\node (0) at (3.1,2.4) {0};
			\node (1) at (3.1,0.8) {1};
			\node (2) at (4.4,-0.5) {2};
			
			\draw (0) -- (1);
			\draw (0) -- (4);
			\draw (1) -- (2);
			\draw (1) -- (4);
			\draw (3) -- (4);
		\end{tikzpicture}
	\end{subfigure}
	\begin{subfigure}[c]{0.5\textwidth}
		\centering
		\begin{quantikz}[
			row sep=0.25cm,
			scale=0.9, transform shape]
			\lstick{$q_0$} & \gate{H} & \gate{Z_\gamma} & \qw          & \qw          & \qw          & \gate{Z_\gamma}& \gate{R_x(\beta)} & \qw\\
			\lstick{$q_1$} & \gate{H} & \ctrl{-1}       & \gate{Z_\gamma} & \gate{Z_\gamma} & \qw     & \qw     & \gate{R_x(\beta)} & \qw \\
			\lstick{$q_2$} & \gate{H} & \qw             & \ctrl{-1}       & \qw          & \qw        & \qw  & \gate{R_x(\beta)} & \qw \\
			\lstick{$q_3$} & \gate{H} & \qw             & \qw          & \qw          & \gate{Z_\gamma}& \qw & \gate{R_x(\beta)} & \qw \\
			\lstick{$q_4$} & \gate{H} & \qw             & \qw          & \ctrl{-3}       & \ctrl{-1}  & \ctrl{-4}     & \gate{R_x(\beta)} & \qw
		\end{quantikz}
	\end{subfigure}
	\caption{Example graph and QAOA-MaxCut circuit for a single layer.}
	\label{fig:eg_QAOA}
\end{figure*}
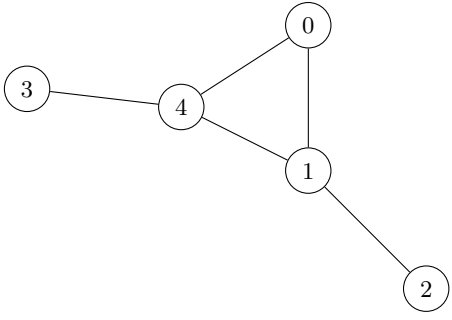

Here, the controlled-$Z_\gamma$ gate is defined as
\[\begin{quantikz}
	&\gate{Z_\gamma}&\qw\\&\ctrl{-1}&\qw
\end{quantikz} : = \begin{quantikz}
	&\targ{X} & \gate{R_z(\gamma)} & \targ{X} & \qw \\
	& \ctrl{-1} & \qw & \ctrl{-1}& \qw
\end{quantikz} \]
Each controlled-$Z_\gamma$ gate is symmetric on the two target qubits and commutes with all other controlled-$Z_\gamma$ gates, allowing flexible rearrangement when multiple gates share qubits. On neutral-atom quantum computers, the gate can be further decomposed into a sequence of native C-Z and single-qubit gates:
\[\begin{quantikz}[row sep=0.25cm,column sep=0.3cm]
	& \gate{H}& \gate{Z} & \gate{H} &\gate{R_z(\gamma)} & \gate{H} & \gate{Z} & \gate{H} &\qw \\
	& \qw & \ctrl{-1} & \qw & \qw & \qw &  \ctrl{-1}& \qw &\qw
\end{quantikz}\]
Consider two controlled-$Z_\gamma$ with one common target qubit, the circuit on these three qubits can be arranged as 
\[\begin{quantikz}[row sep=0.25cm,column sep=0.28cm]
	& \gate{H}& \gate{Z} & \qw & \gate{H} &\gate{R_z(\gamma)} & \gate{H} & \gate{Z} & \qw & \gate{H} &\qw \\
	& \qw & \ctrl{-1} & \ctrl{1} & \qw & \qw & \qw &  \ctrl{-1}& \ctrl{1} & \qw &\qw\\
	& \gate{H}& \qw & \gate{Z} & \gate{H} &\gate{R_z(\gamma)} & \gate{H} &\qw & \gate{Z} & \gate{H} &\qw \\
\end{quantikz}\]
where we use the facts that controlled-$Z_\gamma$ gate is symmetric on the two target qubits and C-Z gates commute. More generally, if a set of controlled-$Z_\gamma$ gates corresponds to a graph with paths no longer than $1$, i.e., a star forest in graph theory, then the corresponding quantum circuit can be arranged in a layered structure as
\begin{equation}\label{eq:sf}
	\begin{quantikz}[row sep=0.25cm,column sep=0.3cm]
		&\gate{\mathcal H} & \gate{\mathcal Z} & \gate{\mathcal H} & \gate{\mathcal R} & \gate{\mathcal H} & \gate{\mathcal Z} & \gate{\mathcal H} & \qw
	\end{quantikz}
\end{equation}
Here, $\mathcal H$, $\mathcal Z$, and $\mathcal R$ denote sets of Hadamard, C-Z, and $R_z$ gates, respectively. Each block can be compiled independently using the algorithms developed in \Cref{sec:methods} 

In summary, our compilation of QAOA-MaxCut circuits proceeds in two stages: (i) the input graph is decomposed into a collection of star forests; (ii) for each star forest, the associated Hadamard, C-Z, and $R_z$ layers are compiled independently. Star-forest decomposition is performed via a greedy edge-ordering heuristic aligned with the qubit layout, as described in \Cref{sec:sf_deco}. The Hadamard layer at the end of a star-forest subcircuit can be merged with the Hadamard layer at the beginning of the subsequent subcircuit to reduce depth. 
\section{Numerical Experiments}
In this section, we evaluate the performance of the proposed compilation methods through numerical experiments. The results demonstrate both the effectiveness of our algorithm in reducing circuit layers and transport operations, and their practical computational efficiency. 

Our experiments consist of five parts. We first compare our single-qubit gate compilation scheme with a na\"ive row-by-row or column-by-column strategy. We then use small-scale instances to empirically verify the $4/3$-optimality for Pauli gate compilation. Next, we evaluate the proposed scheduling algorithm for C-Z gates. We further apply the full compilation pipeline to QAOA-MaxCut circuits to assess its performance. Finally, we measure the runtime of the compilation procedures for both single-qubit and C-Z gates under different array sizes, illustrating their scalability. 
\subsection{Performance Evaluation of Single-Qubit Gate Compilation}
We evaluate the performance of the proposed single-qubit gate compilation methods through numerical experiments. For each gate family, we generate random pattern matrices $M$ over the corresponding group, with matrix dimensions ranging from $n=10$ to $n=200$. For each dimension $n$, $100$ independent instances are sampled.

To reflect locally correlated gate patterns, the random matrices are generated as follows. We first sample a sparse binary matrix with a prescribed density. This pattern is then smoothed using a Gaussian filter to introduce short-range correlations, followed by thresholding to obtain a binary support. Finally, each nonzero entry is mapped uniformly at random to an element of the target gate set, producing the final pattern matrix. Examples are shown in \Cref{fig:randompattern}.  

\begin{figure}
	\centering
	\includegraphics[width=1.0\linewidth]{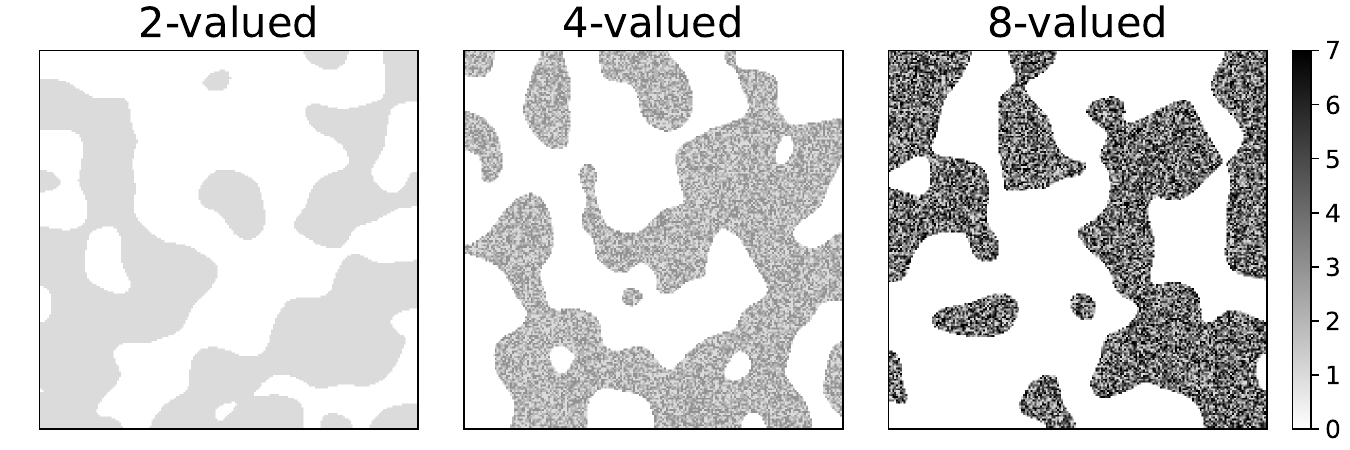}
	\caption{Randomly generated matrix patterns.}
	\label{fig:randompattern}
\end{figure}

We compare our methods with a na\"ive strategy that applies single-qubit gates independently row-by-row or column-by-column. For each instance, we record the total number of required operations. We report the average operation count with one standard deviation, as well as the ratio between our methods and the na\"ive method. We conduct this evaluation for four representative gate families: self-inverse gates, Pauli gates, phase gates, and $\pi/8$ gates. The results are summarized in \Cref{fig:avgrankcomparison}, where the upper panels show the average operation counts and the lower panels show the corresponding ratios.

\begin{figure*}[t]
	\centering
	\includegraphics[width=\textwidth]{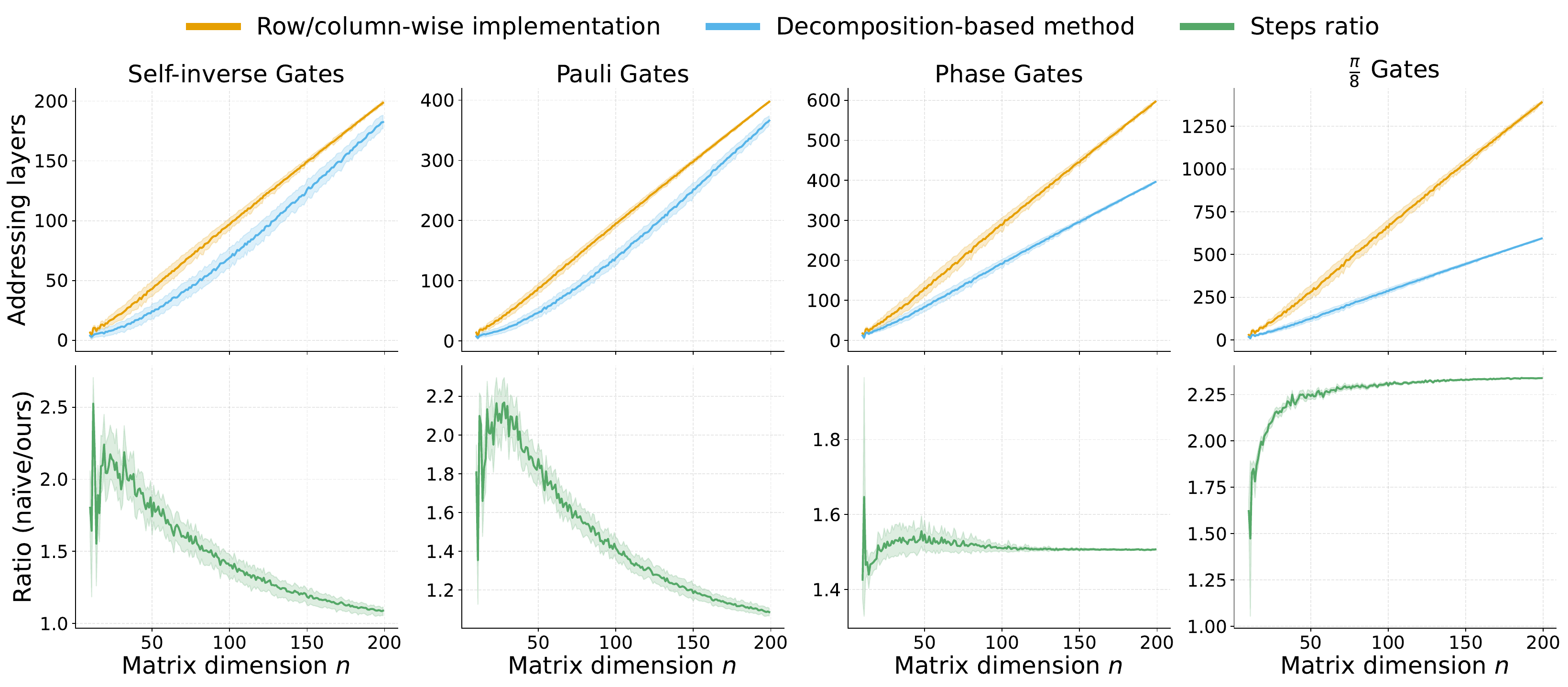}
	\caption{Comparison between our compilation methods and the na\"ive methods for several single-qubit gate families.}
	\label{fig:avgrankcomparison}
\end{figure*}

Across all tested dimensions, our methods outperform the na\"ive approach for self-inverse and Pauli gates, with a diminishing relative gap as the matrix size increases. In contrast, for phase gates and $\pi/8$ gates, the advantage of our methods becomes increasingly pronounced with growing system size.

\subsection{Empirical Verification of the $4/3$-optimality Bound for Pauli Gates}
To empirically validate the theoretical $4/3$-optimality guarantee for Pauli gates compilation, we conduct numerical experiments on small-scale instances. Specifically, for atom arrays of sizes $4\times 4$, $4\times 5$, $5\times 5$, and $5\times 6$, we randomly generate Pauli gates acting on each atom. For each instance, the optimal implementation solution is obtained via exhaustive search and compared against the output of our algorithm by computing the ratio of required operations. 

Since the ratios take values from a small discrete set, we visualize their frequency and cumulative frequency in \Cref{fig:boxplotsratio}. Separate histograms are shown for each array size, illustrating the occurrence of different approximation ratios. Together, these results provide numerical evidence supporting the $4/3$-optimality bound on small instances.
 \begin{figure*}
	\centering
	\includegraphics[width=1.0\linewidth]{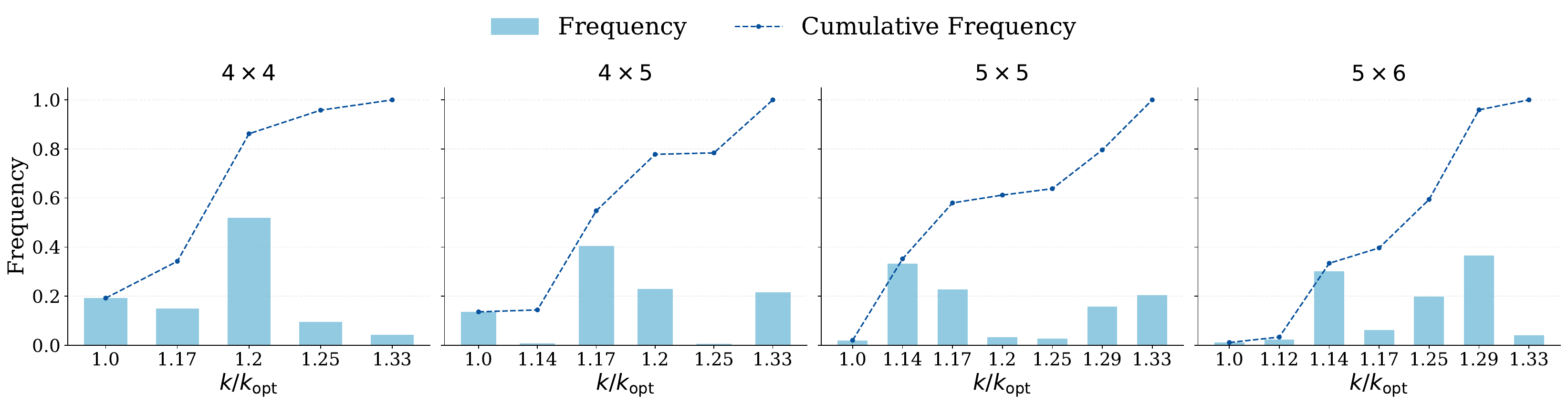}
	\caption{Empirical verification of the $4/3$-optimality bound for Pauli gates. }
	\label{fig:boxplotsratio}
\end{figure*}

\subsection{Performance Evaluation of C-Z Gates Compilation}
In our numerical evaluation of C-Z gate compilation, we consider neutral-atom arrays of size $n\times n$ with $n$ ranging from $10$ to $200$. For each array size, a random set of required C-Z gates is generated by independently including each pair of atoms with probability $8/n^2$. With this choice, the expected number of C-Z gates scales quadratically with the array size $n$, or equivalently linearly with the total number of atoms. This keeps the gate density comparable across different array sizes.
\begin{figure}
	\centering
	\includegraphics[width=1.0\linewidth]{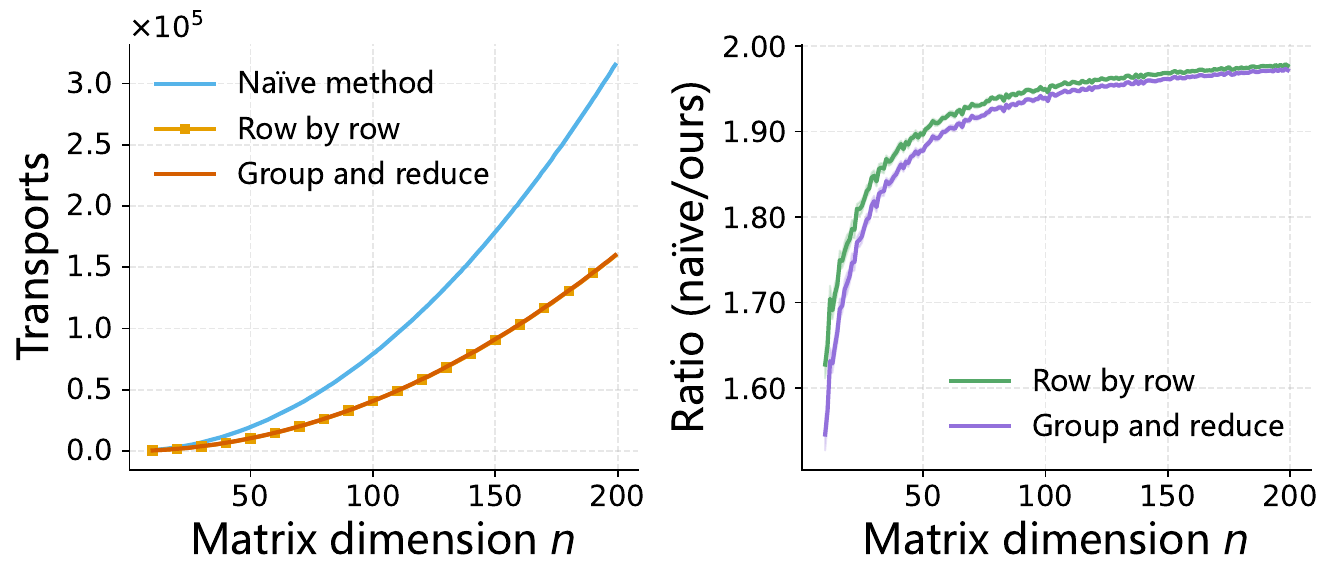}
	\caption{Numerical results of C-Z gates compilation.}
	\label{fig:mcz}
\end{figure}

For each value of $n$, we generate $20$ independent instances and record the average number of transport operations required by our compilation methods, as well as by a na\"{i}ve baseline that implements the required C-Z gates one at a time. The results are summarized in \Cref{fig:mcz}. In particular, we compare the two strategies for compiling row/column-aligned C-Z gates introduced in \Cref{sec:intra_CZ}, labeled as ``row by row" and ``group and reduce" in \Cref{fig:mcz}, respectively. The left panel shows the average number of transports as a function of the array size, while the right panel shows the ratio between the transport cost of the na\"{i}ve baseline and that of our proposed compilation methods. We observe that the na\"{i}ve baseline requires approximately twice as many transport operations as either proposed method, while the two proposed methods exhibit nearly identical performance. A more detailed characterization of the regimes in which each strategy offers a relative advantage is an interesting direction for future work.

\subsection{Performance Evaluation of QAOA-MaxCut Compilation}
To evaluate the performance of our compilation method on structured circuits, we consider QAOA circuits for MaxCut on random graphs. We fix the neutral-atom array size to $30\times 30$ and generate Erdős–Rényi graphs with the number of vertices ranging from $30$ to $500$ and edge probability $\rho = 0.05$. For each graph size, we generate $50$ independent instances to mitigate statistical fluctuations. 

\begin{figure}
	\centering
	\includegraphics[width=\linewidth]{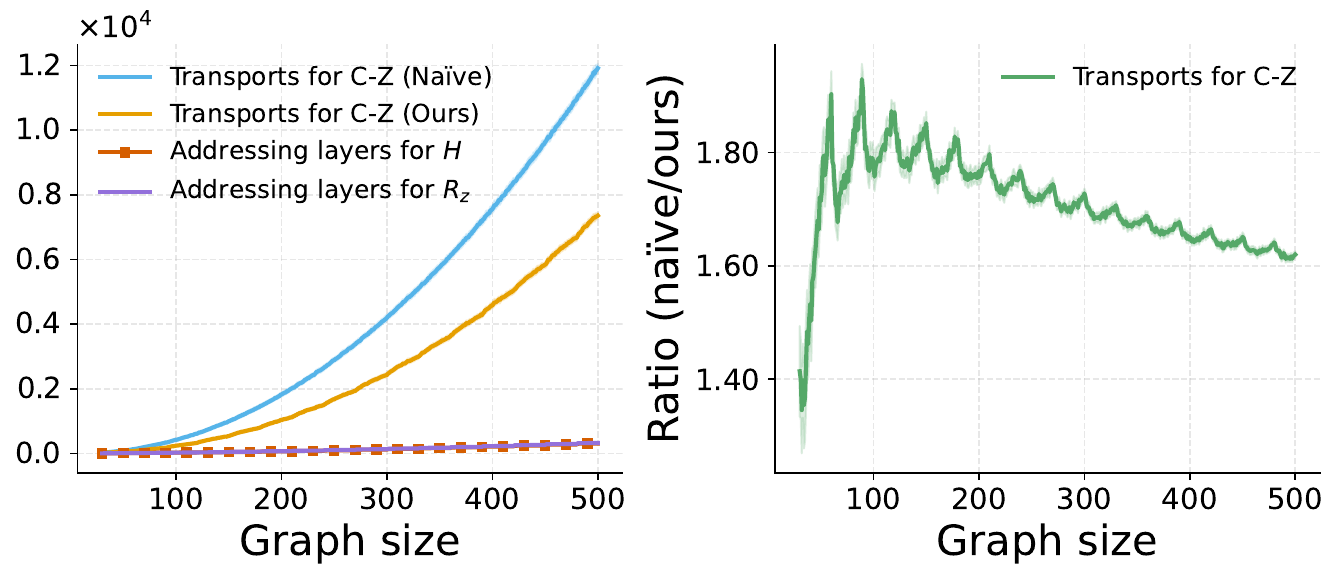}
	\caption{QAOA-MaxCut compilation performance}
	\label{fig:qaoacompile}
\end{figure}

For each instance, we record the total number of atom transports required for implementing C-Z gates, as well as the number of addressing layers required for Hadamard gates and $R_z$ rotations, as produced by our compilation method. For comparison, we also record the corresponding number of atom transports for C-Z gates under the na\"{i}ve method. We report the mean and standard deviation of these quantities over all trials. The results are summarized in \Cref{fig:qaoacompile}. The left panel shows the average number of transports or addressing layers as a function of the graph size, while the right panel shows the ratio between the transport cost of the na\"{i}ve baseline and that of our proposed compilation method. We observe that the overall circuit depth is dominated by the implementation of C-Z gates, and that our compilation method reduces the C-Z transport cost by more than $30\%$ compared with a na\"{i}ve sequential implementation of C-Z gates.

\subsection{Compilation Time Evaluation}
We further evaluate the classical runtime of our compilation algorithms. We consider atom arrays with linear sizes $20, 40, 60, 80, 100, 120, 140, 160, 180$, and $200$, where the array size denotes the array dimension rather than the total number of atoms. For each size, we examine five representative gate families: self-inverse gates, Pauli gates, phase gates, $\pi/8$ gates, and C-Z gates. For each combination of array size and gate family, we generate $100$ independent random instances using the same generation procedure as in the previous experiments. Our compilation methods are applied to each instance, and the compilation time is recorded. We report the mean compilation time over the $100$ samples for each setting, thereby characterizing the scaling behavior of the classical compilation cost across different gate families and system sizes.  
\begin{figure}
	\centering
	\includegraphics[width=1.0\linewidth]{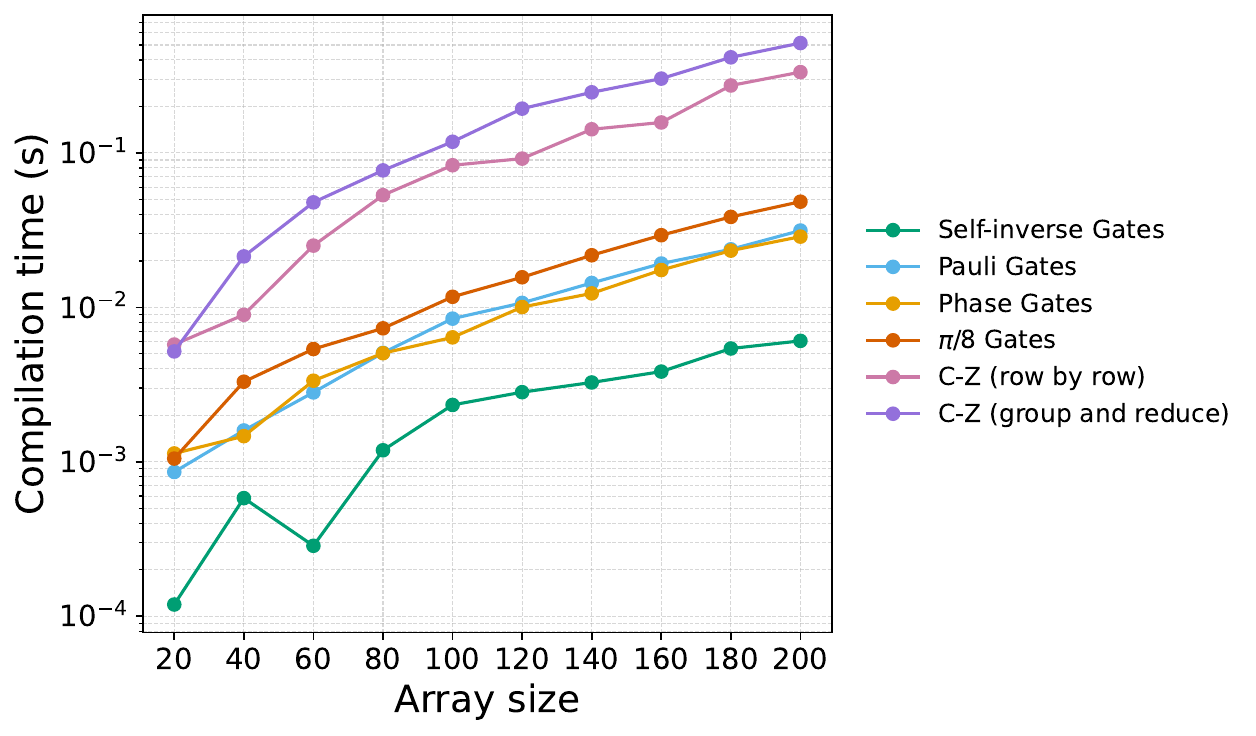}
	\caption{Compilation time evaluation.}
	\label{fig:timeevaluation}
\end{figure}

All experiments were performed on a standard desktop workstation equipped with a 12th Gen Intel Core i7-12700H CPU (2.30 GHz) and 16 GB of RAM, running a 64-bit operating system on an x64 architecture. The results are summarized in \Cref{fig:timeevaluation}. Each curve shows the average compilation time as a function of the atom array size for the compilation method for a fixed gate family. All compilation tasks complete within one second across the tested array sizes, indicating that the proposed methods are computationally efficient and well suited for near-term neutral-atom array scales. 

\section{Conclusion}
In this work, we have presented a comprehensive compilation framework for structured quantum gate families on neutral-atom arrays. For single-qubit gates, we exploit the algebraic structures of gate families at the matrix level, enabling low-rank decompositions over appropriate groups and reducing the number of physical operations. For C-Z gates, we map the transport and adjacency constraints to classical graph-theoretic problems, allowing efficient scheduling using standard algorithms. 

We demonstrated the effectiveness of our methods through extensive numerical experiments. Across representative single-qubit gate families, our compilation methods reduce the number of addressing layers by up to a factor of two compared with naïve row- or column-wise implementations, and small-scale tests confirm the theoretical $\tfrac{4}{3}$-optimality bound for Pauli gates. For C-Z gates, the proposed scheduling strategy reduces the required number of atom transport operations by approximately 50\%. When applied to QAOA circuits for MaxCut, the complete framework reduces the C-Z transport cost by more than 30\% on average. In addition, our timing benchmarks show that all compilation tasks complete within one second fofr the tested array sizes, indicating that the proposed methods are computationally efficient at near-term neutral-atom scales.


These results indicate that our compilation techniques can efficiently handle realistic neutral-atom arrays within feasible runtime limits. Looking forward, our approach provides a foundation for optimizing larger-scale quantum circuits and can be extended to other gate families and quantum algorithms, facilitating more practical and scalable implementations on near-term neutral-atom quantum processors. Our evaluation focuses on addressing layers and atom transports as direct measures of scheduling and shuttling overhead. In future work, it would be useful to incorporate more comprehensive hardware-aware metrics, such as fidelity, decoherence, movement duration, and approximate success probability, as suggested by recent unified evaluation frameworks for neutral-atom compilers~\cite{khusainov2026practicalinsightsfaircomparison}.
\appendix
\section{Auxiliary Lemmas and Proofs}\label{sec:lems&pfs}
\begin{lem}[Lemma 8 in \cite{song2025rank}]\label{lem:mn2_rank}
	Let $\mathcal T = [X_1|X_2]$ be an $m\times n \times 2$ tensor over $\GF(2)$. Then 
	\[\rank(\mathcal T) \ge \frac{\rank(X_1) + \rank(X_2) + \rank(X_1 + X_2)}{2}.\]
\end{lem}
\begin{proof}
	In this proof, addition in $\GF(2)$ is denoted by $\oplus$.	Let $\rank(\mathcal T) = r$, $\rank(X_1) = r_1$, $\rank(X_2) = r_2$, and $\rank(X_1 + X_2) = r_3$. By definition of tensor rank, we have $r_1, r_2\le r \le r_1 + r_2$. Since $\rank(\mathcal T) = r$, there exist matrices $G\in \GF(2)^{m\times r}$, $F\in \GF(2)^{r\times n}$, and $\lambda_1, \dots, \lambda_r$, $\mu_1, \dots, \mu_r\in \GF(2)$ such that 
	\[X_1 = G \begin{bmatrix}\lambda_1\\ & \lambda_2\\ && \ddots \\ &&& \lambda_r\end{bmatrix}F, \quad X_2 = G \begin{bmatrix}\mu_1\\ & \mu_2\\ && \ddots \\ &&& \mu_r\end{bmatrix}F.\]
	By $r_1, r_2\le r$, there are at least $r_1$ ones among $\{\lambda_i\}$ and at least $r_2$ ones among $\{\mu_i\}$. Furthermore, since $r\le r_1 + r_2$, at least $r_1 + r_2 - r$ entries of $\lambda_i \oplus \mu_i$ are zero, implying that at most $2r - (r_1 + r_2)$ entries are one. It then follows that 
	\[X_1 + X_2 = G \begin{bmatrix}\lambda_1 \oplus \mu_1\\ & \lambda_2 \oplus \mu_2\\ && \ddots \\ &&& \lambda_r \oplus \mu_r\end{bmatrix}F,\]
	and hence
	\[r_3 \le 2r - (r_1 + r_2).\]
	Rearranging gives 
	\[2r \ge r_1 + r_2 + r_3,\]
	which concludes the proof.  
\end{proof}

\begin{proof}[Proof of \Cref{lem:rho_rank}]
	There exist invertible matrices $S$ and $T$ over $\GF(2)$ such that $SNT$ is in Smith normal form, namely, 
	\[SNT = \begin{bmatrix}I_r&0\\0&0\end{bmatrix},\]
	where $I_r$ denotes the $r\times r$ identity matrix. Hence, the $\GF(2)$-rank of $N$ is $r$. 
	
	Consider any elementary row transformation matrix $L$ over $\GF(2)$, which must be of the form
	\[L = \begin{bmatrix}
		1\\&\ddots\\ && 1 \\ && v_1& 1\\ && \vdots&&\ddots\\ &&v_t&&&1 
	\end{bmatrix}, \quad \text{where }\; \begin{bmatrix}v_1\\\vdots \\ v_t\end{bmatrix} \in \{0, 1\}^t. \]
	Such a matrix is also invertible over the ring $\mathbb Z_4$. Its inverse is given by
	\[L^{-1} = \begin{bmatrix}
		1\\&\ddots\\ && 1 \\ && \widetilde v_1& 1\\ && \vdots&&\ddots\\ &&\widetilde v_t&&&1 
	\end{bmatrix},\]
	where 
	\[\widetilde v_i = \begin{cases}
		3, & v_i = 1,\\
		0, & v_i = 0,
	\end{cases}\quad \forall i = 1, \dots, t.\]
	Similarly, any elementary column transformation over $\GF(2)$ and any permutation matrix remain invertible over $\mathbb Z_4$. Since the matrices $S$ and $T$ are products of permutation matrices and elementary row or column transformations over $\GF(2)$, they are therefore invertible over $\mathbb Z_4$ as well. 
	
	Applying the row transformation $S$ and column transformation $T$ to $M$ over the ring $\mathbb Z_4$, we obtain 
	\begin{equation}\label{eq:SAT}
		SMT = \begin{bmatrix}D_r& 0\\ 0 & 0\end{bmatrix}, \quad \text{where }\; D_r = I_r \cdot 2.
	\end{equation}
	Since any decomposition of $M$ of the form \eqref{eq:rho_A} induces a decomposition of $SMT$ with the same number of terms, and vice versa due to the invertibility of $S$ and $T$, we have 
	\[\rho_4(M) = \rho_4(SMT). \]
	The diagonal form in \eqref{eq:SAT} immediately yields a decomposition with $r$ rank-one terms, implying
	\[\rho_4(M) = \rho_4(SMT) \le r. \]
	On the other hand, since $D_r$ is a submatrix of $SMT$, we have 
	\begin{equation}\label{eq:ASATD_r}
		\rho_4(M) = \rho_4(SMT) \ge \rho(D_r). 
	\end{equation}
	It therefore suffices to show that $\rho_4(D_r) = r$.  
	
	Indeed, since 
	\[D_r = \sum_{t = 1}^r \left(e_i e_i^\top\right) \cdot 2,\]
	we have $\rho_4(D_r) \le r$. We prove the reverse inequality by induction on $r$. 
	\begin{itemize}
		\item First, consider the case $r = 2$. Suppose $\rho_4(D_2) = 1$. Then there exist $a, b, c, d\in \mathbb Z_4$ such that 
		\[\begin{bmatrix}
			a\\b
		\end{bmatrix}\begin{bmatrix}c & d\end{bmatrix} = \begin{bmatrix}2&0\\0&2\end{bmatrix},\]
		which implies
		\[  ac = 2,\quad
		ad = 0, \quad
		bc = 0, \quad bd = 2.\]
		From the first and fourth equations, we conclude that $a, b, c, d\not = 0$. The second and third equations then force $a = b = c = d = 2$, which contradicts $ac = 2$ and $bd = 2$. Hence, $\rho(D_2) \ge 2$. 
		\item Next, assume that $\rho_4(D_r) \ge r$ but $\rho_4(D_{r+1}) \le r$. Then there exists a decomposition 
		\[D_{r+1} = \sum_{t=1}^r u_t v_t^\top, \quad u_t, v_t \in \mathbb Z^{r+1}_4. \]
		Write 
		\[u_t = \begin{bmatrix}u_{t, 0}\\ \widehat u_t\end{bmatrix},\quad v_t = \begin{bmatrix}v_{t,0} \\ \widehat v_t\end{bmatrix}, \quad  t = 1, \dots, r. \]
		Then 
		\begin{equation}\label{eq:2_D_r}
			\sum_{t=1}^r \begin{bmatrix}
				u_{t, 0}v_{t, 0}& u_{t, 0} \widehat v_t^\top\\
				v_{t, 0}\widehat u_t & \widehat u_t \widehat v_t^\top
			\end{bmatrix} = \begin{bmatrix}2 & 0\\ 0 & D_r\end{bmatrix}.
		\end{equation}
		Since $\sum_{t=1}^r u_{t, 0} v_{t, 0} = 2 \not = 0$, there must be an invertible element among $\{u_{t, 0}, v_{t, 0}\}_{t=1}^r$. Without loss of generality, assume $u_{t_0, 0}$ is invertible. From \eqref{eq:2_D_r}, we have 
		\[\sum_{t=1}^r u_{t, 0} \widehat v_t = 0,\]
		which implies
		\[\widehat v_{t_0} = - u_{t_0, 0}^{-1} \sum_{t\not = t_0} u_{t, 0}\widehat v_t. \]
		Substituting this into \eqref{eq:2_D_r} yields
		\[\begin{aligned}
			D_r &= \sum_{t=1}^r \widehat u_t \widehat v_t^\top\\
			&= \sum_{t\not = t_0} \widehat u_t \widehat v_t^\top + \widehat u_{t_0} \widehat v_{t_0}^\top\\
			&= \sum_{t\not = t_0} \widehat u_t \widehat v_t^\top - \sum_{t\not = t_0} u_{t_0, 0}^{-1} u_{t, 0} \widehat u_{t_0} \widehat v_{t}^\top\\
			&= \sum_{t\not = t_0} \left(\widehat u_t - u_{t_0, 0}^{-1} u_{t, 0} \widehat u_{t_0}\right) \widehat v_t^\top, 
		\end{aligned}\]
		which expresses $D_r$ as a sum of at most $r - 1$ rank-one terms, contradicting the induction hypothesis $\rho_4(D_r)\ge r$. Therefore, $\rho_4(D_{r+1}) \ge r + 1$. 
	\end{itemize}
	By induction, we conclude that $\rho_4(D_r)\ge r$ for all $r\ge 2$. Combined with the upper bound, this yields $\rho_4(D_r) = r$, and consequently $\rho_4(M) = r$, completing the proof. 
\end{proof}
\section{Star-Forest Decomposition}\label{sec:sf_deco}
In our implementation, we employ a simple greedy algorithm to decompose the graph into a collection of star forests. The algorithm iteratively assigns edges to existing groups whenever the star-forest property is preserved; if not, a new group is created.

A key design choice is the order in which edges are processed. Rather than traversing edges arbitrarily, edges are processed in lexicographic order of their endpoints, sorted increasingly first by the smaller endpoint and then by the larger one. This ordering is motivated by the physical layout of neutral-atom arrays, where graph vertices are mapped to qubits row by row according to their indices. By prioritizing edges in this way, the resulting star forests tend to involve qubits concentrated within adjacent rows, which can reduce the number of transport operations required for implementing C-Z gates. 

We emphasize that this ordering does not affect the correctness of the star-forest decomposition; it is purely a heuristic to improve compilation efficiency under hardware constraints. Any valid star-forest decomposition can be used in our compilation framework, and the strategy described here is one concrete choice that performs well in practice. 

\clearpage
\bibliographystyle{apsrev4-2}
\bibliography{refs}

\end{document}